\documentclass[11pt]{article}
\usepackage{fullpage}

\usepackage{times}
\usepackage{comment,amsfonts,amssymb,amsmath,amsthm,graphicx,algorithm,algorithmic}
\newcommand{\commentout}[1]{}

\ifx\pdftexversion\undefined
\usepackage[colorlinks,linkcolor=black,filecolor=black,citecolor=black,urlco
lor=black,pdfstartview=FitH]{hyperref}
\else
\usepackage[colorlinks,linkcolor=blue,filecolor=blue,citecolor=blue,urlcolor
=blue,pdfstartview=FitH]{hyperref}
\fi

\newcommand{\alert}[1]{\textbf{\color{red}
		[[[#1]]]}\marginpar{\textbf{\color{red}**}}\typeout{ALERT:
		\the\inputlineno: #1}}

\def\MathF{\hbox{\rm I\kern-2pt F}}
\def\MathP{\hbox{\rm I\kern-2pt P}}
\def\MathR{\hbox{\rm I\kern-2pt R}}
\def\MathZ{\hbox{\sf Z\kern-4pt Z}}
\def\MathN{\hbox{\rm I\kern-2pt I\kern-3.1pt N}}
\def\MathC{\hbox{\rm \kern0.7pt\raise0.8pt\hbox{\footnotesize I}
		\kern-4.2pt C}}
\def\MathQ{\hbox{\rm I\kern-6pt Q}}

\newcommand{\poly}{{\rm poly}}

\newcommand{\mommit}[1]{}
\newcommand{\namedref}[2]{\hyperref[#2]{#1~\ref*{#2}}}
\newcommand{\sectionref}[1]{\namedref{Section}{#1}}
\newcommand{\appendixref}[1]{\namedref{Appendix}{#1}}

\newcommand{\theoremref}[1]{\namedref{Theorem}{#1}}

\newcommand{\algref}[1]{\namedref{Algorithm}{#1}}
\newcommand{\claimref}[1]{\namedref{Claim}{#1}}

\newcommand{\tableref}[1]{\namedref{Table}{#1}}

\newtheorem{theorem}{Theorem}
\newtheorem{lemma}{Lemma}

\newtheorem{claim}[lemma]{Claim}

\usepackage{pdfsync}
\usepackage{authblk}

\def\tO{\tilde{O}}

\def\eps{\epsilon}
\usepackage{authblk}

\newcommand{\tr}{{\rm trun}}
\begin{document}

\title{Centralized, Parallel, and Distributed Multi-Source Shortest Paths via Hopsets and Rectangular Matrix Multiplication}

\author[1]{Michael Elkin}
\author[1]{Ofer Neiman}

\affil[1]{Department of Computer Science, Ben-Gurion University of the Negev,
Beer-Sheva, Israel. Email: \texttt{\{elkinm,neimano\}@cs.bgu.ac.il}}

\date{}
\maketitle

\begin{abstract}
Consider an undirected weighted graph $G = (V,E,w)$. We study the problem of computing $(1+\eps)$-approximate shortest
paths for $S \times V$, for a subset $S \subseteq V$ of $|S| = n^r$ sources, for some $0 < r \le 1$.
We devise a significantly improved algorithm for this problem in the entire range of parameter $r$, in both the classical centralized and the parallel (PRAM) models of computation, and in a wide range of $r$ in the distributed (Congested Clique) model. Specifically, our centralized algorithm for this problem requires time
$\tilde{O}(|E| \cdot n^{o(1)} + n^{\omega(r)})$, where $n^{\omega(r)}$ is the time required to multiply an $n^r \times n$ matrix by an $n \times n$ one. Our PRAM algorithm has polylogarithmic time $(\log n)^{O(1/\rho)}$, and its work complexity is
$\tilde{O}(|E| \cdot n^\rho + n^{\omega(r)})$, for any arbitrarily small constant $\rho >0$. 

In particular, for $r \le 0.313\ldots$, our centralized algorithm computes $S \times V$ $(1+\eps)$-approximate shortest paths in $n^{2 + o(1)}$ time. Our PRAM polylogarithmic-time algorithm has work complexity $O(|E| \cdot n^\rho + n^{2+o(1)})$, for any arbitrarily small constant $\rho >0$.
Previously existing solutions either require centralized time/parallel work of $O(|E| \cdot |S|)$ or provide  much weaker approximation guarantees.

In the Congested Clique model, our algorithm solves the problem in polylogarithmic time for $|S| = n^r$ sources, for $r \le 0.655$, while previous state-of-the-art algorithms did so only for $r \le 1/2$. Moreover, it improves previous bounds for all $r > 1/2$. For unweighted graphs, the running time is improved further to $\poly(\log\log n)$.

%We also devise efficient algorithms for computing $(1+\eps)$-approximate shortest paths from each vertex to its $k$ nearest neighbors in {\em directed} graphs. Here too the running time is only $O(n^{2+o(1)})$ even for polynomially large $k\le n^{0.168}$.

%Our algorithm combines fast matrix multiplication with hopsets. Closely related ideas were formerly used in the context of the Congested Clique model by Censor-Hillel et al. \cite{CDKL19}. That model, however, suppresses heavy local computations. We show that these computations can be replaced by fast rectangular  matrix multiplication.

\end{abstract}

	\section{Introduction}
	
	We consider the problem of computing {\em $(1+\eps)$-approximate shortest paths} (henceforth, $(1+\eps)$-ASP) in undirected weighted graphs $G = (V,E,w)$, $|V| = n$, for an arbitrarily small $\eps > 0$. We study this problem in the centralized, parallel (PRAM) and distributed (Congested Clique) models of computation.
	Our focus is on computing $(1+\eps)$-ASP for $S \times V$, for a set $S \subseteq V$ of {\em sources}, $|S| = n^r$, for a constant parameter $0 < r \le 1$.
	
	This is one of the most central, fundamental and intensively studied problems in Graph Algorithms. Most of the previous research concentrated on one of the two following scenarios: the {\em single-source} ASP (henceforth, approximate SSSP), i.e., the case $|S| = 1$, and the
	{\em all-pairs} ASP (henceforth, APASP), i.e., the case $S = V$.
	
	We next overview most relevant previous results and our contribution in the centralized model of computation, and then turn to the PRAM and distributed models.
	
	\subsection{Centralized Model}
	
	The classical algorithm of Dijkstra solves {\em exact} SSSP problem in time $O(|E| + n \log n)$ \cite{FT87}. Thorup \cite{T04} refined this bound to $O(|E| + n \log\log n)$ when weights are integers.
	Employing these algorithms for ASP problem for $S \times V$ results in running time of $O(|S|(|E| + n \log\log n))$.
	In the opposite end of the spectrum, Galil and Margalit \cite{GM97}, Alon et al. \cite{AGM97} and Zwick \cite{Z02} showed that one can use fast matrix multiplication (henceforth, FMM) to solve $(1+\eps)$-APASP in time $\tilde{O}(n^\omega)$, where $\omega$ is the matrix multiplication exponent. ($n^\omega$ is the time required to multiply two $n \times n$ matrices. The currently best-known estimate on $\omega$ is
	$\omega < 2.372\dots$ \cite{W12,G14,CW90}.)
	
	By allowing larger approximation factors, one can achieve a running time of $\tilde{O}(n^2)$ for APASP. Specifically, Cohen and Zwick \cite{CZ01} devised an algorithm for 3-APASP with this running time, and Baswana and Kavitha \cite{BK06} refined the approximation ratio to $(2,w)$. The notation $(2,w)$ means that for a vertex pair $(u,v)$, their algorithm provides an estimate with a multiplicative error of 2, and an additive error bounded by the maximal  weight of an edge on some shortest $u-v$ path in the graph.
	
	Cohen \cite{C00}, Elkin \cite{E01}, and Gitlitz and the current authors \cite{EGN19} also showed that one can obtain a $(1+\eps,\beta \cdot w)$-approximation for the ASP problem for $S\times V$ in time $O(|E|\cdot n^\rho + |S| \cdot n^{1+1/\kappa})$, where $\beta = \beta(\eps,\kappa,\rho)$ is a quite large constant (as long as $\eps > 0,\rho > 0, 1/\kappa > 0$ are constant), and $w$ is as in the result of Baswana and Kavitha \cite{BK06}.
	
	However, if one insists on a purely multiplicative error of at most $1+\eps$, for an arbitrarily small constant $\eps > 0$, then for dense graphs ($|E| = \Theta(n^2)$), the best-known running time for ASP for $S \times V$ is $\tilde{O}(\min\{|S|\cdot n^2,n^\omega\})$.
	In the current paper we devise an algorithm that solves the problem  in $\tilde{O}(n^{\omega(r)} + |E|\cdot n^{o(1)})$ time,\footnote{In fact, our result holds for arbitrary $0<\eps<1$, see \theoremref{thm:imp-sources}.} where $\omega(r)$ is the matrix multiplication exponent of {\em rectangular matrix multiplication}. That is, $n^{\omega(r)}$ is the time required to multiply an $n^r \times n$ matrix by an $n \times n$ matrix. Coppersmith \cite{Cop97} showed that for $r \le 0.291$, $\omega(r) \le 2+o(1)$, and Le Gall and Urrutia \cite{GU18} improved this bound further to $r \le 0.313$. Denote $\alpha \ge 0.313$ as the maximal value such that $\omega(\alpha)\le 2+o(1)$. Therefore, our algorithm solves $(1+\eps)$-ASP problem for $S \times V$ in $n^{2+o(1)}$ time, as long as $|S| = O(n^\alpha)$. Moreover, the bound on our running time grows gracefully from $n^{2+o(1)}$ to $n^\omega$, as the number of sources $|S|$ increases from $n^\alpha$ to $n$. When $S = V$, our bound matches the bound of Zwick \cite{Z02}. See \tableref{table:asp1}.
	
	Furthermore, Dor et al. \cite{DHZ00} showed that any $(2-\eps)$-ASP algorithm for $S \times V$ that runs in $T(n)$ time,
	for any positive constant $\eps > 0$ and any  function $T(n)$, translates into an algorithm with running time $T(O(n))$ that multiplies two Boolean matrices with dimensions $|S| \times n$ and $n \times n$. Thus, the running time of our algorithm cannot be improved by
	more than a factor of $n^{o(1)}$ without improving the best-known algorithm for multiplying (rectangular) Boolean matrices.
	
	In terms of edge weights, the situation with our algorithm is similar to that with the algorithm of   Zwick \cite{Z02}. Both algorithms apply directly to graphs with polynomially-bounded edge weights. Nevertheless, both of them can be used in conjunction with the Klein-Sairam's reduction of weights \cite{KS97} to provide the same bounds for graphs with arbitrary weights.
	
	\begin{table*}
		\centering
		\small
		\begin{tabular}{ |l|l|l| }
			\hline
			\# of sources & Our running time  & Previous running time \\
			\hline
			$n^{0.1}$ & $n^{2+o(1)}$ & $n^{2.1}$ \\
			\hline
			$n^{0.2}$ & $n^{2+o(1)}$ & $n^{2.2}$ \\
			\hline
			$n^{0.3}$ & $n^{2+o(1)}$ & $n^{2.3}$ \\
			\hline	
			$n^{0.4}$ & $n^{2.011}$ & $n^{2.373}$ \\
			\hline
			$n^{0.5}$ & $n^{2.045}$ & $n^{2.373}$ \\
			\hline
			$n^{0.6}$ & $n^{2.094}$ & $n^{2.373}$ \\
			\hline
			$n^{0.7}$ & $n^{2.154}$ & $n^{2.373}$ \\
			\hline
			$n^{0.8}$ & $n^{2.222}$ & $n^{2.373}$ \\
			\hline
			$n^{0.9}$ & $n^{2.296}$ & $n^{2.373}$ \\
			\hline
			$n^1$ & $n^{2.373}$ & $n^{2.373}$ \\
			\hline
			
		\end{tabular}
		\caption{Results on $(1+\eps)$-ASP for $S\times V$ in the centralized model for weighted graphs (previous running time is for dense graphs).
		}\label{table:asp1}
	\end{table*}
	
	\subsection{Parallel Model}
	
	The situation in the parallel setting (PRAM) is similar to that in the centralized setting. The first parallel $(1+\eps)$-SSSP algorithm with polylogarithmic time
	(specifically, $(\log n)^{\tilde{O}((\log 1/\rho)/\rho)}$ and $O(|E|\cdot n^\rho)$ work, for any arbitrarily small constant parameter $\rho > 0$, was devised by Cohen \cite{C00}. Her bounds were improved in the last five years by \cite{EN16,EN19hop,L20,ASZ20}, culminating in polylogarithmic time and $\tilde{O}(|E|)$ work \cite{L20,ASZ20}.
	All these aforementioned algorithms are randomized.
	
	On the opposite end of the spectrum, algorithms of Galil and Margalit \cite{GM97}, Alon et al. \cite{AGM97}, and Zwick \cite{Z02} (based on FMM) can be used in the PRAM setting. They give rise to deterministic polylogarithmic time $\tilde{O}(n^\omega)$ work \cite{Z02} for the $(1+\eps)$-APSP problem.
	
	By using sparse spanners, the algorithm of Cohen \cite{C00} in conjunction with that of Baswana and Sen \cite{BS03} provides polylogarithmic time and $O(|E|\cdot n^{1/\kappa} + |S| \cdot n^{1+1/\kappa})$ work for $(2+\eps)\kappa$-ASP for $S \times V$, where $\kappa =  1,2,\dots$ is a parameter. Recently, Gitlitz and the current authors \cite{EGN19} also showed that one can have $(1+\eps,\beta \cdot w)$-ASP for $S \times V$ in polylogarithmic time and $O(|E|\cdot n^\rho + |S|\cdot n^{1+1/\kappa})$ work, where $\beta = \beta(\eps,\kappa,\rho)$ is a large constant (as long as $\eps,\rho,1/\kappa>0$ are constant), and $w$ is as above.
	
	Nevertheless, if one insists on a purely multiplicative error of at most $1+ \eps$, currently best-known solutions for the ASP problem for $S \times V$ that run in polylogarithmic time require work at least
	$\Omega(\min\{|S| \cdot |E|, n^\omega\})$. Our parallel algorithm for the problem with $|S| = n^r$ sources, $0 < r \le 1$, has polylogarithmic time $(\log n)^{O(1/\rho)}$ and work
	$\tilde{O}(n^{\omega(r)} + |E| \cdot n^\rho)$, for any arbitrarily small constant $\rho > 0$. Similarly to the centralized setting, this results in work $n^{2+o(1)} +\tilde{O}(|E|\cdot n^\rho)$, for any arbitrarily small constant $\rho > 0$,  as long as $|S| = O(n^\alpha)$, $\alpha = 0.313$, and it improves Zwick's bound \cite{Z02} of $n^\omega$ (which applies for $(1+\eps)$-APASP) for all values of $r < 1$.
	The aforementioned reduction of \cite{DHZ00} implies that the work complexity of our algorithm cannot be improved by more than a factor of $n^{o(1)}$ without improving he best-known centralized algorithm for multiplying (rectangular) Boolean matrices.
	
	Our algorithm uses FMM and hopsets. The ingredient that builds hopsets is randomized, but by using a new deterministic construction of hopsets from \cite{EM20}, one can make it deterministic, with essentially the same bounds. As a result our ultimate $(1+\eps)$-ASP algorithms (both centralized and parallel ones) become deterministic.

	\subsection{Distributed Model}
	
	In the Congested Clique model, every two vertices of a given $n$-vertex graph $G=(V,E)$, may communicate in each round by a message of $O(\log n)$ bits. The running time of an algorithm is measured by the number of rounds. Computing shortest paths in this model has been extensively studied in the last decade.  An exact APSP algorithm was devised in \cite{CKK15} with running time $O(n^{1 - 2/\omega}) = O(n^{0.158\ldots})$ for unweighted undirected graphs (or with $1+\epsilon$ error in weighted directed graphs), and in $\tilde{O}(n^{1/3})$ time for weighted directed graphs. The latter result was improved in \cite{G16} to $n^{0.209}$ when the weights are constant.
	
	The first algorithm with polylogarithmic time for weighted undirected graphs was devised by \cite{BKKL17}, who showed an ($1+\epsilon$)-approximate {\em single-source} shortest paths algorithm. In \cite{CDKL19}, among other results, a ($1+\epsilon$)-ASP algorithm with polylogarithmic time was shown for a set of $\tilde{O}(n^{1/2})$ sources. For unweighted graphs, the running time was recently improved by \cite{DP20} to $\poly(\log\log n)$, with a similar restriction of $O(n^{1/2})$ sources.
	
	%Our algorithm combines FMM, hopsets, and ideas from \cite{CKK15,CDKL19,DP20}.
	In the current paper we obtain an algorithm for the $(1+\eps)$-ASP in the Congested Clique model for weighted undirected graphs with polylogarithmic time, for a set of
	$|S| = O(n^{{1+ \alpha} \over2}) = O(n^{0.655\ldots})$ sources.
	For larger sets of sources, our  running time gracefully increases until it reaches $\tO(n^{0.158})$ time when $S=V$ (see \tableref{table:asp-CC}).
	Denoting $|S| = n^r$, our algorithm outperforms the state-of-the-art bound of \cite{CDKL19} for all $0.5 < r < 1$.
	In the case of unweighted graphs, we provide a similar improvement over the result of \cite{DP20}: our ($1+\epsilon$)-ASP algorithm has $\poly(\log\log n)$ time, allowing up to $n^{0.655}$ sources.
	
	\begin{table*}[ht]
		\centering
		%\small
		\begin{tabular}{ |l|l|l|l| }
			\hline
			\# of sources & Our running time  & Running time of \cite{CDKL19} & Running Time of \cite{CKK15} \\
			\hline
			$n^{0.5}$ & $\tilde{O}(1)$ & $\tilde{O}(1)$ & $n^{0.158}$ \\
			\hline
			$n^{0.6}$ & $\tilde{O}(1)$ & $n^{0.06}$ & $n^{0.158}$\\
			\hline
			$n^{0.7}$ & $n^{0.006}$ & $n^{0.13}$ & $n^{0.158}$\\
			\hline
			$n^{0.8}$ & $n^{0.04}$ & $n^{0.2}$ & $n^{0.158}$\\
			\hline
			$n^{0.9}$ & $n^{0.1}$ & $n^{0.26}$ & $n^{0.158}$\\
			\hline
			$n^1$ & $n^{0.158}$ & $n^{1/3}$ & $n^{0.158}$\\
			\hline

		\end{tabular}
		\caption{Results on $(1+\eps)$-ASP for $S\times V$ in the Congested Clique model (for any constant $\eps>0$, and hiding constants and lower order terms).
		}\label{table:asp-CC}
	\end{table*}
	
	\subsection{Additional Results}
	
	We also devise an algorithm for the $(1+\eps)$-approximate {\em $k$-nearest neighbors} (henceforth, $k$-NN) problem in PRAM. Here $k$, $1 \le k \le n$, is a parameter. For a vertex $v$, let $z_1,z_2,\ldots$ be all other vertices ordered by their distance from $v$ in non-decreasing order, with ties broken arbitrarily.
	A vertex $u$ is {\em in $(1+\eps)$-approximate $k$-NN of $v$} if it is no farther from $v$ than $(1+\eps) d_G(v,z_k)$.
	The objective is to compute $(1+\eps)$-approximate shortest paths for some set ${\cal P}$ of pairs of vertices,
	that for every vertex $u \in V$ contains at least $k$ pairs $(u,v)$ with $v$ being in the $(1+\eps)$-approximate $k$-NN of $v$.
	%such that $u$ is among the $(1+\eps)$-approximately $k$-NN to $v$ (or vice versa).
	Our algorithm for this problem applies even in {\em directed} weighted graphs. It requires polylogarithmic time and $\tilde{O}(\min\{n^\omega,k^{0.702} n^{1.882} + n^{2+o(1)}\})$ work. For $k = O(n^{0.168})$, this work is $n^{2+0(1)}$, and for $k = o(n^{0.698})$, this bound is better than $n^\omega$, i.e., it improves the bound for $(1+\eps)$-APASP problem.

	%In all the settings discussed above, we can also return a succinct representation of the approximate shortest paths, rather than just report distances. Each path can be reported in time proportional to the number of edges it contains (and within a constant number of rounds and work proportional to the number of edges in PRAM). To achieve this, we employ path-reporting hopsets from \cite{EN16,EN19hop}, in conjunction with witnesses of matrix multiplication from \cite{GM93,Z02}.
	
	\subsection{Technical Overview}\label{sec:tech}
	
	As was mentioned above, our algorithms employ hopsets.
	A graph $H= (V,E',w')$ is a {\em $(1+\eps,\beta)$-hopset} for a graph $G = (V,E,w)$, if for every vertex pair $u,v \in V$, we have
	\begin{equation}
		\label{eq:hopsetdef}
		d_G(u,v) \le d_{G \cup H}^{(\beta)}(u,v) \le (1 + \eps) d_G(u,v)~.
	\end{equation}
	Here $d_{G \cup H}^{(\beta)}(u,v) $ stands for {\em $\beta$-bounded distance} between $u$ and $v$ in $G \cup H$, i.e., the length of the shortest $u-v$ path between them with at most $\beta$ edges (henceforth, {\em $\beta$-bounded path}).
	
	Our algorithm is related to the algorithm of \cite{CDKL19},
	designed for $(1+\eps)$-ASP for $S \times V$ in the distributed Congested Clique (henceforth, CC) model.
	Their algorithm starts with computing a $(1 +\eps,\beta)$-hopset $H$ for the input graph $G$. It then adds $H$ to $G$, and creates an adjacency matrix $A$ of $G \cup H$. It then creates a matrix $B$ of dimensions $|S| \times n$, whose entries $B_{u,v}$, for $(u,v) \in S \times V$, are defined as $w(u,v)$ if $(u,v) \in E$, and $\infty$ otherwise.
	Then compute distance products $B \star A, (B \star A) \star A,\ldots,(B \star A^{\beta-1}) \star A = B \star A^\beta$.
	By equation (\ref{eq:hopsetdef}), $B \star A^\beta$ is a $(1+\eps)$-approximation of all distances in $S \times V$.
	
	Censor-Hillel et al. \cite{CDKL19} developed an algorithm for efficiently multiplying {\em sparse} matrices in the distributed CC model. They view the matrices $B, B \star A,\ldots,B \star A^{\beta-1}$, as sparse square $n \times n$ matrices, and as a result compute $B \star A^{\beta}$ efficiently via their (tailored to the CC model) algorithm.
	In particular, their algorithm does not use Strassen-like fast matrix multiplication (FMM) techniques, but rather focuses on carefully partitioning all the products that need to be computed in a naive matrix product of dimensions $|S|\times n$ by $n \times n$ among $n$ available processors.
	
	Our first observation is that this product can be computed much faster using best available fast {\em rectangular} matrix multiplication (FRMM) algorithms. This observation leads to our $(1+\eps)$-ASP algorithms for weighted graphs that significantly improve the state-of-the-art in all the three computational models that we consider
	(the centralized, PRAM, and distributed CC).
	In the centralized and PRAM models we also  need to convert matrix distance products into ordinary algebraic matrix products.
	This is, however, not difficult, and was accomplished, e.g., in \cite{Z02}. We employ the same methodology (of \cite{Z02}).
	In the distributed CC model there is no need in this conversion, because the processors are assumed to possess unbounded computational capacity. (Indeed, the algorithm of \cite{CDKL19} works directly with distance products.)
	We needed, however, to implement fast {\em rectangular} MM in this model. In \cite{CKK15} fast MM of {\em square} matrices was implemented in the distributed CC model. We argue that the same approach is applicable for fast rectangular MM as well. This leads to our improved $(1+\eps)$-ASP algorithms in the distributed CC model (cf. Table \ref{table:asp-CC}).

	Remarkably, while so far hopsets were used extensively in parallel/distributed/dynamic/streaming settings \cite{C00,B09,N14,HKN14,HKN16,EN16,EN19hop,CDKL19}, there were no known applications of hopsets in the classical centralized setting. Our results demonstrate that this powerful tool is extremely useful in the classical setting as well.

	\subsection{Organization}
	
	After reviewing some preliminary results in \sectionref{sec:prel}, we describe our algorithm for $(1+\epsilon)$-ASP for $S\times V$ in the standard centralized model in \sectionref{sec:MSP}. In \sectionref{sec:CC} we show our algorithm for $(1+\epsilon)$-ASP for $S\times V$ in the Congested Clique model that substantially improves the number of allowed sources while maintaining polylogarithmic time (and $\poly(\log\log n)$ for unweighted graphs). In \appendixref{sec:pram} we show our PRAM algorithm for $(1+\epsilon)$-ASP for $S\times V$, and in \appendixref{sec:knn} we show our PRAM algorithm for approximate (and exact) distances to $k$-NN.

	\section{Preliminaries}\label{sec:prel}

	\paragraph{Matrix Multiplication and Distance Product.} Fix an integer $n$. For $0\le r\le 1$,  let $w(r)$ denote the exponent of $n$ in the number of algebraic operations required to compute the product of an $n^r\times n$ matrix by an $n\times n$ matrix.
	
	Let $1 \le s,q\le n$.
	Let $A$ be an $s\times n$ matrix. We denote the entry in row $i$ and column $j$ of the matrix $A$ by $A_{ij}$. The transpose of $A$ is $A^T$. We use * to denote a wildcard, e.g., the notation $A_{*j}$ refers to the vector which is the $j$-th column of $A$.
	For an $n\times q$ matrix $B$, define the distance product $C=A\star B$ by
	\[
	C_{ij}=\min_{1\le k\le n}\{A_{ik}+B_{kj}\}~,
	\]
	for $1\le i\le s$ and $1\le j\le q$.
	We say that $C'$ is a $\delta$-approximation to $C$ if for all $i,j$,
	$C_{ij}\le C'_{ij}\le \delta\cdot C_{ij}$.

	The following theorem is implicit in \cite{Z02}. We will provide a sketch of the proof since we would like to apply it in parallel setting (and also for completeness).
	\begin{theorem}[\cite{Z02}]\label{lem:star}
		Let $M,R$ be positive integers. Let $A$ be an $n^r\times n$ matrix and $B$ an $n\times n$ matrix, whose entries are all in $\{1,...,M\}\cup\{\infty\}$. Then there is an algorithm that computes an ($1+\frac{1}{R}$)-approximation to $A\star B$ in time $\tilde{O}(R\cdot n^{w(r)}\cdot\log M)$.
	\end{theorem}
	\begin{proof}[Sketch]
		It was shown in \cite{AGM97} that the distance product $C=A\star B$ can be computed by defining $\hat{A}_{ij}=(n+1)^{M-A_{ij}}$ and similarly $\hat{B}_{ij}$. Then $C$ can be derived from $\hat{C}=\hat{A}\cdot \hat{B}$ by $C_{ij}=2M-\lfloor\log_{n+1}\hat{C}_{ij}\rfloor$. Since the values of entries in the matrices $\hat{A}$ and $\hat{B}$ are quite large, each algebraic operation (when computing the standard product $\hat{A}\cdot \hat{B}$) will take $O(M\log n)$ time. So the running time will be $\tilde{O}(M\cdot n^{w(r)})$.
		
		In order to reduce the dependence on the maximal weight $M$ to logarithmic, one can apply scaling: for each $k=0,1,\dots,\log M-\log R$, define $A'$ by setting \[
		A'_{ij}=\left\{\begin{array}{ccc}\lceil A_{ij}/2^k\rceil & A_{ij}\le R\cdot 2^k\\\infty&\textrm{otherwise}\end{array}\right.
		\]
		and similarly define $B'$. Then compute $C'=A'\star B'$ by the above method. \cite{Z02} showed that a ($1+\frac{1}{R}$)-approximation to $A\star B$ can be obtained by taking the scaled up minimum (entry-wise) of all these $C'$. The point is that each $A'$ and $B'$ have entries in $\{1,\dots,R\}\cup\{\infty\}$. Hence the running time is indeed $\tilde{O}(R\cdot n^{w(r)}\cdot\log M)$ (the factor $O(\log M)$ comes from the number of different choices of $k$).
	\end{proof}	
	
	\paragraph{Witnesses.} Given an $s\times n$ matrix $A$ and an $n\times q$ matrix $B$, an $s\times q$ matrix $W$ is called a {\em witness} for $C=A\star B$ if for all $i,j$, $C_{ij}=A_{iW_{ij}}+B_{W_{ij}j}$. It was shown in \cite{GM93,Z02} how to compute the matrix $W$ in almost the same time required to compute $C$ (up to logarithmic factors). This holds also for a witness for $C'$ which is a $c$-approximation for $C$ (see \cite[Section 8]{Z02}), for some $c \ge 1$.	The witness can assist us in recovering the actual paths, rather than just reporting distance estimates. Since computing witnesses is done by an appropriate distance product, these witnesses can also be efficiently computed in the PRAM model.
	
	\paragraph{Hopsets.}
	%Given a weighted undirected graph $G=(V,E,w)$, let $d_G$ be the shortest path metric of $G$. A graph $H=(V,E')$ is called a $(\epsilon,\beta)$-{\em hopset} for $G$ if for all $u,v\in V$
	%\[
	%d_G(u,v)\le d_{G\cup H}^{(\beta)}(u,v)\le(1+\epsilon)\cdot d_G(u,v)~,
	%\]
	%where $d_{G\cup H}^{(\beta)}$ is the $\beta$-bounded distance in the union graph $G\cup H$, i.e., the length of the shortest path containing at most $\beta$ edges.
	
	Recall the definition of hopsets in the beginning of \sectionref{sec:tech}.
	A randomized construction of hopsets was gives in \cite{C00}, see also \cite{N14,HKN16,EN16}. The following version was shown in \cite{EN19hop}.
	
	\begin{theorem}[\cite{EN19hop}]\label{thm:hopset-trade}
		For any weighted undirected graph $G = (V, E)$ on $n$ vertices and parameter $\kappa>1$, there is a randomized algorithm running in time  $\tilde{O}(|E|\cdot n^{1/\kappa})$, that computes $H$ of size $O(n^{1+1/\kappa})$, which is an $(\epsilon,\beta)$-hopset (for every $0<\epsilon<1$ simultaneously) with $\beta =\left(\frac{\kappa}{\epsilon}\right)^{O(\kappa)}$.
	\end{theorem}
	We note that a forthcoming paper \cite{EM20} provides a  {\em deterministic} construction of hopsets with similar properties.
	There are two differences, which have essentially no effect on our result. First, the hopbound in
	\cite{EM20} is  $\beta =\left(\frac{\log n}{\epsilon}\right)^{O(\kappa)}$. Second, the construction there accepts $\eps > 0$ as a part of its input.
	Nevertheless, their hopsets can be used to make our results in PRAM deterministic, with essentially the same parameters. Our centralized algorithm can also be made deterministic using a hopset construction from \cite{HKN16}.
	
	%\begin{theorem}[\cite{HKN16}]\label{thm:hopset}
	%For any weighted undirected graph $G = (V, E)$ on $n$ vertices with weights in $\{1,...,M\}$ and a parameter $0<\epsilon<1$,
	%there is a deterministic algorithm that runs in time $O(|E|/\epsilon^{\tilde{O}(\sqrt{\log n})}\cdot\log M)$, and computes an $(\epsilon,\beta)$-hopset with $n/\epsilon^{\tilde{O}(\sqrt{\log n})}\cdot\log M$ edges where $\beta = 1/\epsilon^{\tilde{O}(\sqrt{\log n})}$.
	%\end{theorem}
	%

	\paragraph{Eliminating Dependence on Aspect Ratio.} The aspect ratio of a graph $G$ is the ratio between the largest to smallest edge weight. A well-known reduction by \cite{KS97} asserts that to compute ($1+\epsilon$)-approximate shortest paths in $G=(V,E)$ with $|V|=n$, it suffices to compute ($1+\epsilon$)-approximate shortest paths in a collection of at most $\tilde{O}(|E|)$ graphs $\{G_t\}$. The total number of (non-isolated) vertices in all these graphs is $O(n\log n)$, the total number of edges is $\tilde{O}(|E|)$, and the aspect ratio of each graph is $O(n/\epsilon)$. This reduction can be performed in parallel (PRAM EREW) within $O(\log^2n)$ rounds and work $O(|E|)$. Thus it can also be done in the standard centralized model in $\tilde{O}(|E|)$ time. See also \cite[Section 4]{EN16} for more details.
	%We conclude that if there exists an algorithm that computes ($1+\epsilon$)-approximate shortest paths for a graph $G$ with $n$ vertices and aspect ratio $M$, whose running time is $O(f(n,M))$, then the running time can be made $\tilde{O}(f(n,n))$ (as long as $f$ is convex, which is always the case here). 
	Since in our algorithms the dependence on the aspect ratio will be logarithmic, in all that follows we may assume $M=\poly(n)$.

	\section{Multi-Source Shortest Paths}\label{sec:MSP}
	
	Let $G=(V,E,w)$ be a weighted undirected graph
	%, with integer weights in $\{1,...,M\}$
	and fix a set of $s$ sources $S\subseteq V$. %\footnote{Even if the graph has real weights with aspect ratio $M$, we can obtain integers by multiplying all weights by $1/\epsilon$ and rounding to the nearest integer. This will affect the running time by at most a $\log(1/\epsilon)$ factor.}
	We compute a ($1+\epsilon$)-approximation for all distances in $S\times V$, by executing \algref{alg:MSP}.

	\begin{algorithm}[ht]
		\caption{$\texttt{ASP}(G,S,\epsilon)$}\label{alg:MSP}
		\begin{algorithmic}[1]
			\STATE Let $H$ be an $(\epsilon,\beta)$-hopset for $G$;
			\STATE Set $R=\beta/\epsilon$;
			\STATE Let $A$ be the adjacency matrix of $G\cup H$;
			\STATE Let $B^{(1)}=A_{S*}$;
			\FOR {$t$ from $1$ to $\beta-1$}
			\STATE Let $B^{t+1}$ be a $(1+1/R)$-approximation to $B^{(t)}\star A$;
			\ENDFOR
			\RETURN $B^{(\beta)}$;
		\end{algorithmic}
	\end{algorithm}

	The first step is to compute an $(\epsilon,\beta)$-hopset $H$, for a parameter $\kappa\ge 1$ with $\beta=\left(\frac{\kappa}{\epsilon}\right)^{O(\kappa)}$ as in \theoremref{thm:hopset-trade}. Let $A$ be the adjacency matrix of $G\cup H$ and fix $R=\beta/\epsilon$. For every integer $1\le t\le \beta$, let $B^{(t)}$ be an $s\times n$ matrix such that for all $i\in S$ and $j\in V$, $B^{(t)}_{ij}$ is a $(1+\frac{1}{R})^{t-1}$-approximation to $d_{G\cup H}^{(t)}(i,j)$. Note that $B^{(1)}=A_{S*}$ is a submatrix of $A$ containing only the rows corresponding to the sources $S$. %In addition, it is easy to see that the maximal distance in $G\cup H$ is $O(Mn)$.
	
	The following claim asserts that taking an approximate distance product of $B^{(t)}$ with the adjacency matrix yields $B^{(t+1)}$.
	
	\begin{claim}\label{claim:star}
		Let $c,c'\ge 1$. Let $A$ be the adjacency matrix of an $n$-vertex graph $G=(V,E)$, and let $B$ be an $s\times n$ matrix (whose rows correspond to $S\subseteq V$) so that for all $i,j$, $B_{ij}$ is a $c$-approximation to $d_G^{(t)}(i,j)$, for some positive integer $t$. Let $C=B\star A$ and $C'$ be a $c'$-approximation to $C$.  Then, for all $i,j$,  $C'_{ij}$ is a $c\cdot c'$-approximation to $d_G^{(t+1)}(i,j)$.
	\end{claim}
	
	\begin{proof}
		Consider a pair of vertices $i\in S$ and $j\in V$. By definition of the $\star$ operation, $C_{ij}=\min_{1\le k\le n}\{B_{ik}+A_{kj}\}$.
		Let $\pi$ be the shortest path in $G$ from $i$ to $j$ that contains at most $t+1$ edges, and let $k\in V$ be the last vertex before $j$ on $\pi$. Since $B_{ik}$ is a $c$-approximation to $d_G^{(t)}(i,k)$ and $A_{kj}$ is the edge weight of $\{k,j\}$, we have that $B_{ik}+A_{kj}$ is a $c$-approximation to $d_G^{(t+1)}(i,j)$.
		The assertion of  the claim follows since $C_{ij}\le C'_{ij}\le c'\cdot C_{ij}$.
	\end{proof}
	
	Given $B^{(t)}$, we compute $B^{(t+1)}$ as a ($1+\frac{1}{R}$)-approximation to $B^{(t)}\star A$. Using \theoremref{lem:star} this can be done within $\tilde{O}(R\cdot n^{w(r)})$ rounds. Thus, the total running time to compute $B^{(\beta)}$ is
	\[
	O(\beta\cdot R\cdot n^{w(r)})=n^{w(r)}\cdot(\kappa/\epsilon)^{O(\kappa)}
	\]

	By \claimref{claim:star}, $B^{(\beta)}$ is a $(1+\frac{1}{R})^{\beta-1}\le e^\epsilon=1+O(\epsilon)$ approximation to $d_{G\cup H}^{(\beta)}(u,v)$ for all $u\in S$ and $v\in V$. Since $H$ is an $(\epsilon,\beta)$-hopset, $B^{(\beta)}$ is a $(1+O(\epsilon))$-approximation to $d_G(u,v)$, for all $u \in S$, and $v \in V$.
	
	%We summarize with the following theorem.
	
	%\begin{theorem}\label{thm:sources}
	%Let $G=(V,E)$ be a weighted undirected graph, fix $S\subseteq V$ of size $n^r$ for some $0\le r\le 1$, and let $0<\epsilon<1$. Then there is a deterministic algorithm that computes a ($1+\epsilon$)-approximation to all distances in $S\times V$ that runs in time
	%\[
	%n^{w(r)}/\epsilon^{\tilde{O}(\sqrt{\log n})}~.
	%\]
	%\end{theorem}

	\paragraph{Reporting paths.}
	For each approximate distance in $S\times V$ we can also report a path in $G$ achieving this distance. To this end, we compute witnesses for each approximate distance product, and as in \cite[Section 5]{Z02} there is an algorithm that can report, for any $u,v\in V$, a path in $G\cup H$ of length at most $(1+\epsilon)\cdot d_{G \cup H}(u,v)$.
	In order to translate this to a path in $G$, we need to replace the hopset edges by corresponding paths in $G$. We use the fact that the hopsets of \cite{EN19hop} have a path reporting property. That is, each hopset edge of weight $W'$ has a corresponding path $\pi$ of length $W'$ in $G$, and every vertex on $\pi$ stores its neighbors on the path. Thus, we can obtain a $u-v$ path in $G$ in time proportional to its number of edges.
	
	%\alert{do we want to write this for pram as well? We can't return all these paths, just a succinct representation, as in Zwick's paper.}
	
	With conclude with the following theorem.
	\begin{theorem}\label{thm:imp-sources}
		Let $G=(V,E)$ be a weighted undirected graph, fix $S\subseteq V$ of size $n^r$ for some $0\le r\le 1$, and let $0<\epsilon<1$. Then for any $\kappa\ge 1$, there is a deterministic algorithm that computes a ($1+\epsilon$)-approximation to all distances in $S\times V$ that runs in time
		\[
		\tilde{O}(\min\{n^{w(r)}\cdot\left(\kappa/\epsilon\right)^{O(\kappa)},|E|\cdot n^{1/\kappa}\})~.
		\]
		Furthermore, for each pair in $S\times V$, a path achieving the approximate distance can be reported in time proportional to the number of edges in it.
	\end{theorem}

	One may choose $\kappa$ as a slowly growing function of $n$, e.g. $\kappa=(\log\log n)/\log\log\log n$, so that $\kappa^\kappa\le\log n$ and $n^{1/\kappa}=n^{o(1)}$, and obtain running time $\tilde{O}(n^{\omega(r)}+|E|\cdot n^{o(1)})$ (for a constant $\eps > 0$).
	%\alert{are we cheating by hiding $\epsilon^{-\kappa}$ term?}
	%Note that whenever $\epsilon=2^{-o(\sqrt{\log n})}$ is not too small, the running time is $n^{w(r)+o(1)}$.
	We stress that for all $r\le 0.313$, a result of \cite{GU18} gives that $w(r)=2+o(1)$. So even for polynomially large set of sources $S$, with size up to  $n^{0.313}$, our algorithm computes ($1+\epsilon$)-approximate distances  $S\times V$ in time $n^{2+o(1)}$. In fact, for all $r<1$, our bound  improves the current bound
	for $(1+\eps)$-APASP \cite{Z02}.
	%on the exponent of $n\times n$ matrix multiplication, which is $\omega=2.372\dots$.

	%\paragraph{Improved running time.} We can improve the $\epsilon^{\tilde{O}(\sqrt{\log n})}$ factor in the running time whenever $r>0.313$ or when the graph is not extremely dense. To this end, we will use the following randomized construction of hopsets, that allows a tradeoff between the hopbound $\beta$ and the running time.

	Observe that if $r>0.313$, then we can choose $\kappa$ as a large enough constant, so that the running time to compute the hopset, which is $\tilde{O}(|E|\cdot n^{1/\kappa})$, is dominated by $n^{w(r)}$. Alternatively, if $|E|\le n^{2-\delta}$ we may choose $\kappa=1/\delta$, so the running time to compute the hopset will be $\tilde{O}(n^2)=\tilde{O}(n^{w(r)})$ for all $0\le r\le 1$. In both cases we obtain $\beta=(1/\epsilon)^{O(1)}$, so our algorithm to compute $(1+\epsilon)$-approximate shortest paths for $S\times V$ will have running time
	$
	\tilde{O}(n^{w(r)}/\epsilon^{O(1)})
	$.

	%\alert{I'm not sure we should keep both theorems 3,5.  Also, I wrote the version for general k since I didn't like imposing the conditions ($|E|$ is small or $r>0.313$) in the theorem.}

	\section{Improved ASP for $S\times V$ in the Congested Clique Model}\label{sec:CC}
	
	%\alert{define congested clique and put these results in abstract/intro.}

	In this section we show how to improve the ($1+\epsilon$)-ASP for $S\times V$ results of \cite{CDKL19} and \cite{DP20} in the Congested Clique model. Specifically, we show that given a weighted graph $G=(V,E)$ and a set of $S\subseteq V$ sources of size $|S|=n^r$, there is an $\poly(\log n)$ time algorithm to compute ($1+\epsilon$)-ASP for $S\times V$ as long as $r<(1+\alpha)/2\approx 0.655$.  For unweighted graphs, we obtain an improved running time of $\poly(\log\log n)$.
	More generally, for  $S$ of arbitrary size, $|S| = n^r$, the running time is given by
	$\tilde{O}(n^{f(r)})$, where the function $f(r)$ grows from 0 to $1 - {2 \over \omega} \approx 0.158$. (See Table \ref{table:asp-CC} for more details.)
	
	A polylogarithmic running time (respectively, $\poly(\log\log n)$ time for unweighted graphs), was obtained only for $r\le 1/2$ in \cite{CDKL19} (resp., \cite{DP20}). More generally their running time for arbitrary $S$ is $\tilde{O}(\frac{|S|^{2/3}}{n^{1/3}})$.
	
	To achieve these improvements, we use the method of \cite{CKK15} combined with fast rectangular matrix multiplication. We start by devising the fast rectangular MM algorithm in the Congested Clique, summarized in the following theorem.
	
	Whenever $A$ is an $n^r\times n$ matrix, it will be convenient to think of it as an $n\times n$ matrix with the last rows containing only zeros.
	
	\begin{theorem}\label{thm:CCRMM}
		Let $G=(V,E)$ be an $n$-vertex graph, and fix $0<r\le 1$. Let $A$ be an $n^r\times n$ and $B$ an $n\times n$ matrices with entries in $\{1,2,...,M\}$, so that each $x\in V$ holds a unique row of $A$ and of $B$. Then for any $R\ge 1$ there is a deterministic algorithm in the Congested Clique that computes $(1+1/R)$-approximation to $A\star B$ in $O(R\cdot n^{1-2/\omega(r')}\cdot\log M)$ rounds, where $r'$ is the solution to the equation:
		\[
		r' = 1-(1-r)\cdot\omega(r')~.
		\]
		(Recall that $\omega(r')$ is the exponent for $n^{r'}\times n$ MM.)
	\end{theorem}
	\begin{proof}
		For a clearer exposition, we will show the details for $r<(1+\alpha)/2$, and sketch the case of  larger $r$.
		
		Recall that all known fast matrix multiplication algorithms are {\em bilinear}. Given a $q^r\times q$ matrix $S$ and a $q\times q$ matrix $T$, if the algorithm multiplies the matrices $S,T$ in time $q^\sigma$, then it first computes $2m=O(q^\sigma)$ linear combinations $S^{(1)},\dots, S^{(m)}, T^{(1)},\dots, T^{(m)}$. That is, for $1\le x\le m$,
		\begin{equation}\label{eq:S}
			S^{(x)} = \sum_{i=1}^{q^r}\sum_{j=1}^q\alpha_{ij}^{(x)}S_{ij},
		\end{equation}
		and
		\[
		T^{(x)} = \sum_{i=1}^q\sum_{j=1}^q\beta_{ij}^{(x)}T_{ij}.
		\]
		Then the result $U=ST$ is given by
		\[
		U_{ij}=\sum_{x=1}^m\gamma_{ij}^{(x)}S^{(x)}\cdot T^{(x)},
		\]
		where the coefficients $\alpha_{ij}^{(x)},\beta_{i,j}^{(x)},\gamma_{i,j}^{(x)}$ are defined by the algorithm being used.
		
		\paragraph{The case $r< (1+\alpha)/2$.}
		In order to multiply the $n^r\times n$ matrix $A$ by the $n\times n$ matrix $B$, we think of $A$ as an $n^{r-1/2}\times n^{1/2}$ matrix over the ring of $n^{1/2}\times n^{1/2}$ matrices, and similarly $B$ is an $n^{1/2}\times n^{1/2}$ matrix over that ring. Since $r < (1+\alpha)/2$, we get $n^{r-1/2}\le (n^{1/2})^\alpha$. So the fast rectangular matrix multiplication algorithm of \cite{GU18} involves $O((n^{1/2})^2)=O(n)$ multiplications of $n^{1/2}\times n^{1/2}$ matrices. For convenience we assume there are at most $n$ of those, and assign each such product $A^{(x)}\cdot B^{(x)}$ to a vertex $x\in V$. (Note that $A^{(x)}$ is a $n^{1/2}\times n^{1/2}$ matrix, which is a linear combination of the $n^{1/2}\times n^{1/2}$ submatrices of $A$, as defined in \eqref{eq:S}.)

		Let $q=\sqrt{n}$ (assume that this is an integer)\footnote{Up to constant factors, integrality issues have no effect on the complexity of our algorithms.}, and index each integer in $[n]$ by a pair in $[q]\times[q]$. For instance, each vertex $x\in V$ will be associated with the pair $(x_1,x_2)\in [q]\times[q]$. Recall that each $x\in V$ is responsible for computing the product of the $q\times q$ matrices $A^{(x)}\cdot B^{(x)}$.
		We assumed that initially each vertex $x$ knows its rows in $A$ and in $B$, i.e., $A_{xi}$ and $B_{xi}$ for all $i$.
		Here $x$ is a number in $[n]$.
		\begin{itemize}
			\item Each $x=(x_1,x_2)\in V$ sends, for each $y_2\in [q]$, the $2q$ elements $A_{x(*,y_2)}$ and $B_{x(*,y_2)}$ to vertex $(x_2,y_2)$. So vertex $y=(y_1,y_2)$ receives the $q\times q$ matrices $A_{(*,y_1)(*,y_2)}$ and $B_{(*,y_1)(*,y_2)}$.\footnote{As was defined in the preliminaries, the notation $A_{x(*,y_2)}$ stands for the set of entries $A_{x(y_1,y_2)}$ for all possible $y_1\in [q]$. Other notation that involves asterisks is defined analogously.}
			
			\item Each vertex $y=(y_1,y_2)\in V$ computes a  linear combination for each $x\in V$,
			\[
			A^{(x)}_y=\sum_{i,j\in [q]}\alpha_{ij}^{(x)}A_{(i,y_1)(j,y_2)},
			\]
			and similarly for $B$. Then $y$ sends the results to $x$. So vertex $x\in V$ receives these $2n$ sums, and can infer $A^{(x)}$ and $B^{(x)}$. (Note that each $y$ sends $x$ a single entry $A^{(x)}_{y_1y_2}$ of the $q\times q$ matrix $A^{(x)}$.)
			
			\item Each $x\in V$ locally computes the $q\times q$ matrix $C^{(x)}=A^{(x)}\cdot B^{(x)}$, and sends the entry $C^{(x)}_{y_1y_2}$ to vertex $y=(y_1,y_2)$. Now each $y\in V$ has the corresponding entry of all the $n$ product matrices
			$\{C^{(x)} \mid x \in V\}$.  So it can compute, for any $i,j\in[q]$, the $(i,y_1),(j,y_2)$ entry of the output matrix $C$ by
			\[
			C_{(i,y_1)(j,y_2)}=\sum_{x\in V}\gamma_{ij}^{(x)}C^{(x)}_{y_1y_2}~.
			\]
			
			%\alert{ME: define notation with asterisks.}
		\end{itemize}
		
		%(Note that this is the entry $(y_1,y_2)$ of the block matrix $C_{ij}$, for $i,j \in [q]$.)
		Now the entries of $C$ can be redistributed among the $n$ vertices, so that every vertex will know all entries of its respective row.

		\paragraph{The general case.} The case when $r\ge (1+\alpha)/2$ is done in a similar manner. We think of $A$ as a $d^{r'}\times d$ matrix over the ring of $n/d\times n/d$ matrices, for a suitable choice of $d=n^{1/\omega(r')}$ and $r'$, and $B$ as a  $d\times d$ matrix over the same ring. Then the algorithm of \cite{GU18} will take $O(d^{\omega(r')})=O(n)$ multiplcations of $n/d\times n/d$ matrices, each will be assigned to a vertex of $V$. Thus, each vertex will need to receive $(n/d)^2$ elements (those in its $n/d\times n/d$ matrix). The number of rounds required for this is $O(n/d^2)$ (using the dissemination and aggregation of linear combinations as in \cite{CKK15}, the simple case of which was described above).
		
		A quick calculation shows that as $A$ has $n^r$ rows, that are divided into blocks of size $n/d \times n/d$, we need $d^{r'}=n^r/(n/d)=d \cdot n^{r-1}$, or $r'\cdot \log d= \log d+(r-1)\log n$, so
		\[
		r' = 1-\frac{(1-r)\log n}{\log d}~.
		\]
		Plugging in the choice of $d=n^{1/\omega(r')}$, we get an equation for $r'$ :
		\begin{equation}\label{eq:r}
			r' = 1-(1-r)\cdot\omega(r')~.
		\end{equation}
		We conclude that $A\cdot B$ can be computed in $O(n^{1-2/\omega(r')})$ rounds, for $r'$ solving \eqref{eq:r}. To compute an ($1+1/R$)-approximation to $A\star B$, simply apply the reduction of \theoremref{lem:star} (computing each product sequentially locally).
		
	\end{proof}

	\subsection{ASP for $S\times V$ in Weighted Graphs}
	Here we apply the improved rectangular MM to ASP for $S\times V$, using the method of \cite{CDKL19}. For completeness we sketch it below.
	The following theorem was shown in \cite{CDKL19}, based on a construction from \cite{EN19hop}. It provides a fast construction of a hopset with logarithmic hopbound for the Congested Clique model.
	\begin{theorem}[\cite{CDKL19}]
		Let $0<\epsilon<1$.
		For any $n$-vertex weighted undirected graph $G = (V,E)$,
		there is a deterministic construction of a $(1+\epsilon,\beta)$-hopset $H$ with
		$\tilde{O}(n^{3/2})$ edges and $\beta=O(\log n/\epsilon)$, that requires $O(\log^2n/\epsilon)$ rounds in the Congested Clique model.
	\end{theorem}
	Now, we can approximately compute $\beta$-bounded distances in the graph $G\cup H$, by letting $B$ be the adjacency matrix of $G\cup H$, and $A^{(1)}$ the $|S|\times n$ matrix of sources.
	(Specifically, for every pair $(u,v) \in S \times V$, the entry $A_{u,v}^{(1)}$ contains $\omega((u,v))$ if $(u,v) \in E$, and $\infty$ otherwise.)
	Define $A^{(t+1)}=A^{(t)}\star B$, and by \claimref{claim:star} and the definition of hopset, $A^{(\beta)}_{ij}$ is a $(1+\epsilon)\cdot(1+1/R)$-approximation to $d_G(i,j)$ for any $i\in S$ and $j\in V$. Each product is computed by \theoremref{thm:CCRMM} within $\tilde{O}(R\cdot n^{1-2/\omega(r')}\cdot\log M)$ rounds. Taking $R=\lceil 1/\epsilon\rceil$ and recalling that $\beta=O(\log n/\epsilon)$, yields the following.
	
	\begin{theorem}\label{thm:CC-w}
		Given any $n$-vertex weighted undirected graph $G=(V,E)$ with polynomial weights, parameters $0<r<1$, $0<\epsilon<1$, and a set $S\subseteq V$ of $n^r$ sources, let $r'$ be the solution to \eqref{eq:r}. Then there is a deterministic algorithm in the Congested Clique that computes ($1+\epsilon$)-ASP for $S\times V$ within $\tilde{O}(n^{1-2/\omega(r')}/\epsilon)$ rounds.
	\end{theorem}
	
	In particular, for a constant $\epsilon > 0$, when $r<(1+\alpha)/2\approx 0.655$ the running time is $\tilde{O}(1)$. For $r=0.7$, the solution is slightly smaller than $r'=0.4$, for which $\omega(r')\approx 2.01$, so $d=n^{0.497}$, and the number of rounds is $O(n/d^2)\approx O(n^{0.006})$.
	When $r=0.8$, the solution is roughly $r'=0.59$, for which $\omega(r')\approx 2.085$, so $d=n^{0.48}$, and the number of rounds is $O(n/d^2)\approx O(n^{0.04})$. We show a few more values in the following \tableref{table:asp-CC}. (Note that at $r=1$ we converge to the result of \cite{CKK15} for APASP.)
	
	%\alert{This table should also be recomputed, as Table 1 changed. OFER: I think I computed those values from Le Gall, can double check later on..}

	\subsection{ASP for $S\times V$ in Unweighted Graphs}

	In this section we show an improved algorithm for unweighted graphs, based on \cite{DP20}. They first devised a fast algorithm for a sparse emulator: we say that $H=(V,F)$ is an {\em $(\alpha,\beta)$-emulator} for a graph $G=(V,E)$ if for all $u,v\in V$, $d_G(u,v)\le d_H(u,v)\le \alpha\cdot d_G(u,v)+\beta$.
	\begin{lemma}[\cite{DP20}]
		For any $n$-vertex unweighted graph $G =(V, E)$ and $0 <\epsilon < 1$, there is a randomized algorithm in the Congested Clique model that computes $(1+\epsilon,\beta)$-emulator $H$ with $O(n\log\log n)$ edges within $O(\log^2\beta/\epsilon)$ rounds w.h.p., where $\beta = O(\log\log n/\epsilon)^{\log\log n}$.
	\end{lemma}
	Since the emulator is so sparse, all vertices can learn all of its edges within $O(\log\log n)$ rounds. Thus every pair of distance larger than $\beta/\epsilon$ already has an $1+O(\epsilon)$ approximation, just by computing all distances in $H$ locally. It remains to handle distances at most  $\beta/\epsilon$.
	
	The next tool is a bounded-distance hopset that "takes care" of small distances. We say that $H'=(V,E')$ is a $(1+\epsilon,\beta',t)$-hopset if for every pair $u,v\in V$ with $d_G(u,v)\le t$ we have the guarantee of \eqref{eq:hopsetdef}.
	\begin{theorem}[\cite{DP20}]
		There is a randomized construction of a $(1+\epsilon,\beta',t)$-hopset $H'$ with $O(n^{3/2}\log n)$ edges and
		$\beta'=O(\log t/\epsilon)$ that takes $O(\log^2t/\epsilon)$ rounds w.h.p. in the Congested Clique model.
	\end{theorem}
	We take a $(1+\epsilon,\beta',t)$-hopset $H'$ for $G$ with $t= \beta/\epsilon=O(\log\log n/\epsilon)^{\log\log n}$, so that $\beta'=\poly(\log\log n/\epsilon)$.
	As before we let $B$ be the adjacency matrix of $G\cup H'$, and $A^{(1)}$ be the $|S|\times n$ matrix of sources. Define $A^{(s+1)}=A^{(s)}\star B$. By \claimref{claim:star} and the definition of bounded-distance hopset, $A^{(\beta')}_{ij}$ is a $(1+\epsilon)\cdot(1+1/R)$-approximation to $d_G(i,j)$ for any $i\in S$ and $j\in V$ with $d_G(i,j)\le t$. Each distance product is computed by \theoremref{thm:CCRMM} within $\tilde{O}(R\cdot n^{1-2/\omega(r')}\cdot\log M)$ rounds.
	We note that since $G$ is unweighted, the maximal entry in $B$ and any $A^{(s)}$ is $t\cdot (1+\epsilon)$ (we can simply ignore entries of larger weight, i.e. replacing them by $\infty$, since they will not be useful for approximating distances at most $t$). So we have $\log M=\poly(\log\log n)$.
	In the current setting we assume $r \le {{1 + \alpha} \over 2} \approx 0.655$, and so $\omega(r')=2$. We conclude with the following theorem.
	\begin{theorem}
		Given any $n$-vertex unweighted undirected graph $G=(V,E)$, any $0<\epsilon<1$, and a set $S\subseteq V$ of at most $O(n^{0.655\ldots})$ sources, there is a randomized algorithm in the Congested Clique that w.h.p. computes ($1+\epsilon$)-ASP for $S\times V$ within $\poly(\log\log n/\epsilon)$ rounds.
	\end{theorem}

	\section{Acknowledgements}
	
	We are grateful to Fran\c{c}ois Le Gall for explaining us certain aspects of the algorithm of \cite{GU18}.
	
	\bibliographystyle{alpha}
	\bibliography{hopset}

\newcommand{\etalchar}[1]{$^{#1}$}
\begin{thebibliography}{CKK{\etalchar{+}}15}

\bibitem[AGM97]{AGM97}
Noga Alon, Zvi Galil, and Oded Margalit.
\newblock On the exponent of the all pairs shortest path problem.
\newblock {\em J. Comput. Syst. Sci.}, 54(2):255--262, 1997.

\bibitem[ASZ20]{ASZ20}
Alexandr Andoni, Clifford Stein, and Peilin Zhong.
\newblock Parallel approximate undirected shortest paths via low hop emulators.
\newblock 2020.
\newblock to appear in STOC.

\bibitem[Ber09]{B09}
Aaron Bernstein.
\newblock Fully dynamic {(2} + epsilon) approximate all-pairs shortest paths
  with fast query and close to linear update time.
\newblock In {\em 50th Annual {IEEE} Symposium on Foundations of Computer
  Science, {FOCS} 2009, October 25-27, 2009, Atlanta, Georgia, {USA}}, pages
  693--702, 2009.

\bibitem[BK06]{BK06}
Surender Baswana and Telikepalli Kavitha.
\newblock Faster algorithms for approximate distance oracles and all-pairs
  small stretch paths.
\newblock In {\em FOCS}, pages 591--602, 2006.

\bibitem[BKKL17]{BKKL17}
Ruben Becker, Andreas Karrenbauer, Sebastian Krinninger, and Christoph Lenzen.
\newblock Near-optimal approximate shortest paths and transshipment in
  distributed and streaming models.
\newblock In {\em 31st International Symposium on Distributed Computing, {DISC}
  2017, October 16-20, 2017, Vienna, Austria}, pages 7:1--7:16, 2017.

\bibitem[BS03]{BS03}
S.~Baswana and S.~Sen.
\newblock A simple linear time algorithm for computing a $(2k-1)$-spanner of
  ${O}(n^{1+1/k})$ size in weighted graphs.
\newblock In {\em Proceedings of the 30th International Colloquium on Automata,
  Languages and Programming}, volume 2719 of {\em LNCS}, pages 384--396.
  Springer, 2003.

\bibitem[CDKL19]{CDKL19}
Keren Censor{-}Hillel, Michal Dory, Janne~H. Korhonen, and Dean Leitersdorf.
\newblock Fast approximate shortest paths in the congested clique.
\newblock In Peter Robinson and Faith Ellen, editors, {\em Proceedings of the
  2019 {ACM} Symposium on Principles of Distributed Computing, {PODC} 2019,
  Toronto, ON, Canada, July 29 - August 2, 2019}, pages 74--83. {ACM}, 2019.

\bibitem[CKK{\etalchar{+}}15]{CKK15}
Keren Censor{-}Hillel, Petteri Kaski, Janne~H. Korhonen, Christoph Lenzen, Ami
  Paz, and Jukka Suomela.
\newblock Algebraic methods in the congested clique.
\newblock In {\em Proceedings of the 2015 {ACM} Symposium on Principles of
  Distributed Computing, {PODC} 2015, Donostia-San Sebasti{\'{a}}n, Spain, July
  21 - 23, 2015}, pages 143--152, 2015.

\bibitem[Coh00]{C00}
Edith Cohen.
\newblock Polylog-time and near-linear work approximation scheme for undirected
  shortest paths.
\newblock {\em J. {ACM}}, 47(1):132--166, 2000.

\bibitem[Cop97]{Cop97}
Don Coppersmith.
\newblock Rectangular matrix multiplication revisited.
\newblock {\em J. Complex.}, 13(1):42--49, 1997.

\bibitem[CW90]{CW90}
Don Coppersmith and Shmuel Winograd.
\newblock Matrix multiplication via arithmetic progressions.
\newblock {\em J. Symb. Comput.}, 9(3):251--280, 1990.

\bibitem[CZ01]{CZ01}
Edith Cohen and Uri Zwick.
\newblock All-pairs small-stretch paths.
\newblock {\em J. Algorithms}, 38(2):335--353, 2001.

\bibitem[DHZ00]{DHZ00}
D.~Dor, S.~Halperin, and U.~Zwick.
\newblock All-pairs almost shortest paths.
\newblock {\em SIAM J. Comput.}, 29:1740--1759, 2000.

\bibitem[DP20]{DP20}
Michal Dory and Merav Parter.
\newblock Exponentially faster shortest paths in the congested clique.
\newblock In {\em Proceedings of the 39th Symposium on Principles of
  Distributed Computing}, PODC '20, page 59–68, New York, NY, USA, 2020.
  Association for Computing Machinery.

\bibitem[EGN19]{EGN19}
Michael Elkin, Yuval Gitlitz, and Ofer Neiman.
\newblock Almost shortest paths and {PRAM} distance oracles in weighted graphs.
\newblock {\em CoRR}, abs/1907.11422, 2019.

\bibitem[Elk01]{E01}
M.~Elkin.
\newblock Computing almost shortest paths.
\newblock In {\em Proc. 20th ACM Symp. on Principles of Distributed Computing},
  pages 53--62, 2001.

\bibitem[EM20]{EM20}
Michael Elkin and Shaked Matar.
\newblock Deterministic pram approximate shortest paths in polylogarithmic time
  and near-linear work.
\newblock 2020.
\newblock manuscript.

\bibitem[EN19a]{EN16}
Michael Elkin and Ofer Neiman.
\newblock Hopsets with constant hopbound, and applications to approximate
  shortest paths.
\newblock {\em {SIAM} J. Comput.}, 48(4):1436--1480, 2019.

\bibitem[EN19b]{EN19hop}
Michael Elkin and Ofer Neiman.
\newblock Linear-size hopsets with small hopbound, and constant-hopbound
  hopsets in {RNC}.
\newblock In {\em The 31st {ACM} on Symposium on Parallelism in Algorithms and
  Architectures, {SPAA} 2019, Phoenix, AZ, USA, June 22-24, 2019.}, pages
  333--341, 2019.

\bibitem[FT87]{FT87}
Michael~L. Fredman and Robert~Endre Tarjan.
\newblock Fibonacci heaps and their uses in improved network optimization
  algorithms.
\newblock {\em J. {ACM}}, 34(3):596--615, 1987.

\bibitem[Gal14]{G14}
Fran{\c{c}}ois~Le Gall.
\newblock Powers of tensors and fast matrix multiplication.
\newblock In Katsusuke Nabeshima, Kosaku Nagasaka, Franz Winkler, and
  {\'{A}}gnes Sz{\'{a}}nt{\'{o}}, editors, {\em International Symposium on
  Symbolic and Algebraic Computation, {ISSAC} '14, Kobe, Japan, July 23-25,
  2014}, pages 296--303. {ACM}, 2014.

\bibitem[Gal16]{G16}
Fran{\c{c}}ois~Le Gall.
\newblock Further algebraic algorithms in the congested clique model and
  applications to graph-theoretic problems.
\newblock In Cyril Gavoille and David Ilcinkas, editors, {\em Distributed
  Computing - 30th International Symposium, {DISC} 2016, Paris, France,
  September 27-29, 2016. Proceedings}, volume 9888 of {\em Lecture Notes in
  Computer Science}, pages 57--70. Springer, 2016.

\bibitem[GM93]{GM93}
Zvi Galil and Oded Margalit.
\newblock Witnesses for boolean matrix multiplication and for transitive
  closure.
\newblock {\em J. Complex.}, 9(2):201--221, 1993.

\bibitem[GM97]{GM97}
Zvi Galil and Oded Margalit.
\newblock All pairs shortest distances for graphs with small integer length
  edges.
\newblock {\em Inf. Comput.}, 134(2):103--139, 1997.

\bibitem[GU18]{GU18}
Francois~Le Gall and Florent Urrutia.
\newblock Improved rectangular matrix multiplication using powers of the
  coppersmith-winograd tensor.
\newblock In Artur Czumaj, editor, {\em Proceedings of the Twenty-Ninth Annual
  {ACM-SIAM} Symposium on Discrete Algorithms, {SODA} 2018, New Orleans, LA,
  USA, January 7-10, 2018}, pages 1029--1046. {SIAM}, 2018.

\bibitem[HKN14]{HKN14}
Monika Henzinger, Sebastian Krinninger, and Danupon Nanongkai.
\newblock Decremental single-source shortest paths on undirected graphs in
  near-linear total update time.
\newblock In {\em 55th {IEEE} Annual Symposium on Foundations of Computer
  Science, {FOCS} 2014, Philadelphia, PA, USA, October 18-21, 2014}, pages
  146--155, 2014.

\bibitem[HKN16]{HKN16}
Monika Henzinger, Sebastian Krinninger, and Danupon Nanongkai.
\newblock A deterministic almost-tight distributed algorithm for approximating
  single-source shortest paths.
\newblock In {\em Proceedings of the Forty-eighth Annual ACM Symposium on
  Theory of Computing}, STOC '16, pages 489--498, New York, NY, USA, 2016. ACM.

\bibitem[HP98]{HP98}
Xiaohan Huang and Victor~Y. Pan.
\newblock Fast rectangular matrix multiplication and applications.
\newblock {\em J. Complex.}, 14(2):257--299, 1998.

\bibitem[KS97]{KS97}
Philip~N. Klein and Sairam Subramanian.
\newblock A randomized parallel algorithm for single-source shortest paths.
\newblock {\em J. Algorithms}, 25(2):205--220, 1997.

\bibitem[Li20]{L20}
Jason Li.
\newblock Faster parallel algorithm for approximate shortest path.
\newblock 2020.
\newblock to appear in STOC.

\bibitem[Nan14]{N14}
Danupon Nanongkai.
\newblock Distributed approximation algorithms for weighted shortest paths.
\newblock In {\em Symposium on Theory of Computing, {STOC} 2014, New York, NY,
  USA, May 31 - June 03, 2014}, pages 565--573, 2014.

\bibitem[SV81]{SV81}
Yossi Shiloach and Uzi Vishkin.
\newblock Finding the maximum, merging, and sorting in a parallel computation
  model.
\newblock {\em J. Algorithms}, 2(1):88--102, 1981.

\bibitem[Tho04]{T04}
Mikkel Thorup.
\newblock Integer priority queues with decrease key in constant time and the
  single source shortest paths problem.
\newblock {\em J. Comput. Syst. Sci.}, 69(3):330--353, 2004.

\bibitem[Wil12]{W12}
Virginia~Vassilevska Williams.
\newblock Multiplying matrices faster than coppersmith-winograd.
\newblock In Howard~J. Karloff and Toniann Pitassi, editors, {\em Proceedings
  of the 44th Symposium on Theory of Computing Conference, {STOC} 2012, New
  York, NY, USA, May 19 - 22, 2012}, pages 887--898. {ACM}, 2012.

\bibitem[YZ05]{YZ05}
Raphael Yuster and Uri Zwick.
\newblock Fast sparse matrix multiplication.
\newblock {\em {ACM} Trans. Algorithms}, 1(1):2--13, 2005.

\bibitem[Zwi02]{Z02}
Uri Zwick.
\newblock All pairs shortest paths using bridging sets and rectangular matrix
  multiplication.
\newblock {\em J. {ACM}}, 49(3):289--317, 2002.

\end{thebibliography}
	
	\appendix

	\section{PRAM Approximate Multi-Source Shortest Paths}\label{sec:pram}
	
	The algorithm of \sectionref{sec:MSP} can be translated to the PRAM model. In this model, multiple processors are connected to a single memory block, and the operations are performed in parallel by these processors in synchronous rounds. The {\em running time} is measured by the number of rounds, and the {\em work} by the number of processors multiplied by the number of rounds.
	
	To adapt our algorithm to this model, we need to show that approximate distance products can be computed efficiently in PRAM. The second ingredient is a parallel algorithm for hopsets. For the latter, the following theorem was shown in \cite{EN19hop}.
	A deterministic analogue of it was recently shown in \cite{EM20}.
	
	\begin{theorem}[\cite{EN19hop}]\label{thm:hopset-pram}
		For any weighted undirected graph $G = (V, E)$ on $n$ vertices and parameters $\kappa\ge 1$ and $0<\epsilon<1$,	there is a randomized algorithm that runs in parallel time  $\left(\frac{\log n}{\epsilon}\right)^{O(\kappa)}$ and work $\tilde{O}(|E|\cdot n^{1/\kappa})$, that computes an $(\epsilon,\beta)$-hopset with $O(n^{1+1/\kappa}\cdot\log^*n)$ edges where $\beta = \left(\frac{\kappa}{\epsilon}\right)^{O(\kappa)}$.
	\end{theorem}
	
	\paragraph{Matrix multiplication in PRAM.}
	
	Essentially all the known fast matrix multiplication algorithms are based on Strassen's approach of divide and conquer, and thus are amenable to parallelization \cite{HP98}. In particular, these algorithms which classically require time $T(n)$, can be executed in the PRAM (EREW) model within $O(\log^2n)$ rounds and $\tilde{O}(T(n))$ work.
	
	The algorithm of \theoremref{lem:star} boils down to $O(\log M)=O(\log n)$ standard matrix multiplications (albeit the matrices have large entries, with $O(R\log n)$ bits). Thus, we can compute an ($1+\frac{1}{R}$)-approximate distance products of an $n^r\times n$ matrix by an $n\times n$ matrix in $O(R\cdot\poly(\log n))$ rounds and $\tilde{O}(R\cdot n^{w(r)})$ work. %Recalling that we apply this algorithm $\beta$ times, choose $k$ as a small constant so that $\beta=1/\epsilon^{O(1)}$.
	
	The path-reporting mechanism can be adapted to PRAM, by running the \cite{Z02} algorithm sequentially. Since we have only $\beta$ iterations, the parallel time will be only $O(\beta)$ (which is a constant independent of $n$). Once we got the path in $G\cup H$, we can expand all the hopset edges in parallel. We thus have the following result.
	
	\begin{theorem}\label{thm:pram-sources}
		Let $G=(V,E)$ be a weighted undirected graph, fix $S\subseteq V$ of size $n^r$ for some $0\le r\le 1$, and let $0<\epsilon<1$. Then for any $\kappa\ge 1$, there is a randomized parallel algorithm that computes a ($1+\epsilon$)-approximation to all distances in $S\times V$, that runs in $\left(\frac{\log n}{\epsilon}\right)^{O(\kappa)}$ parallel time, using work
		\[
		\tilde{O}(\min\{n^{w(r)}\cdot(\kappa/\epsilon)^{O(\kappa)},|E|\cdot n^{1/\kappa}\})~.
		\]
		Furthermore, for each pair in $S\times V$, a path achieving the approximate distance can be reported within parallel time $(\kappa/\epsilon)^{O(\kappa)}$, and work proportional to the number of edges in it.
	\end{theorem}
	
	Note that we can set $\kappa$ to be an arbitrarily large constant, and obtain a polylogarithmic time and work
	$\tO(n^{\omega(r)} + |E|n^{1/\kappa})$.
	
	%\alert{pram path report, check MPC}
	
	%As above, if $r>0.313$ or $|E|\le n^{2-\delta}$, we can choose a constant $k$ for the hopset, which yields time $\poly((\log n)/\epsilon)$ and work $\tilde{O}(n^{w(r)}\cdot(\log M)/\epsilon^{O(1)})$.

	\section{Approximate Distances to $k$-Nearest Neighbors in PRAM}\label{sec:knn}
	
	%\alert{Define approximate $k$-NN. Perhaps $\tilde{N}_k(v)$ - an approximate set of $k$ nearest neighbors.}
	
	In this section, given a weighted {\em directed} graph $G=(V,E)$, we focus on the task of approximately computing the distances from each $v\in V$ to its $k$ nearest neighbors. The main observation is that we work with rather sparse matrices, since for each vertex we do not need to store distances to vertices that are not among its $k$ nearest neighbors.
	
	In \cite{YZ05} fast algorithms for sparse matrix multiplication were presented. Recall that $\alpha\in[0,1]$ is the maximal exponent so that the product of an $n\times n^\alpha$ by $n^\alpha\times n$ matrices can be computed in $n^{2+o(1)}$ time. Currently by \cite{GU18}, $\alpha\ge 0.313$. Let $\gamma = \frac{\omega-2}{1-\alpha}$.
	
	\begin{theorem}[\cite{YZ05}]\label{thm:YZ}
		The product of two $n\times n$ matrices each with at most $m$ nonzeros can be computed in time
		\[
		\min\{O(n^\omega),m^{\frac{2\gamma}{\gamma+1}}\cdot n^{\frac{2-\alpha\gamma}{\gamma+1}+o(1)}+n^{2+o(1)}\}~.
		\]
	\end{theorem}
	
	We present the following adaptation to distance products in the PRAM model.
	In our setting, a matrix will be sparse if it contains few non-infinity values.
	\begin{lemma}\label{lem:pram-mult}
		For $R\ge 1$, the $(1+\frac{1}{R})$-approximate distance product of two $n\times n$ matrices each with at most $m$ non-infinities can be computed in parallel time $O(R\log^{O(1)}n)$ and work
		\begin{equation}\label{eq:work1}
			\tilde{O}(R\cdot  \min\{n^\omega,m^{0.702}\cdot n^{1.18}+n^{2+o(1)}\})~.
		\end{equation}
	\end{lemma}
	
	%\alert{Here too, use fast parallel multiplication of long integers.}
	\begin{proof}
		The $(1+\frac{1}{R})$-approximate distance product of \theoremref{lem:star} involves $O(\log n)$ standard matrix multiplications. These multiplications can be done in parallel, and we need to compute entry-wise minimum of these matrices. This can also be done very efficiently in PRAM (See e.g., \cite{SV81}). By the reduction described in the proof of \theoremref{lem:star}, the resulting matrices will have $O(m)$ nonzeros (and entries of size $O(n^R)$), so the parallel time to compute each such multiplication is $O(R\log^{O(1)}n)$. Using the currently known bounds on $\omega$ and $\alpha$, we have $\gamma\approx 0.542$. Plugging this in \theoremref{thm:YZ}, the work required is as in \eqref{eq:work1}.
		
	\end{proof}

	%\begin{equation}\label{eq:work}
	%\tilde{O}(R\cdot\log M\cdot  \min\{n^\omega,k^{0.702}\cdot n^{1.882}+n^{2+o(1)}\})~.
	%\end{equation}
	
	\begin{comment}
	Define $S(k,M)$ as the work required (with $O(\log^2n)$ time) to compute the distance product of two such matrices.
	By the reduction of \cite{AGM97}, such distance product reduces to standard matrix multiplication with $O(kn)$ nonzeros (and entries of size $O(n^M)$). Thus, using the currently known bounds on $\omega$ and $\alpha$, by Theorem \ref{thm:YZ}
	\[
	S(k,M)\le \tilde{O}(M\cdot (k^{0.702}\cdot n^{1.882}+n^{2+o(1)}))~.
	\]
	We also note that with $O(\log^2n)$ time one can trivially compute each entry in the distance product using only $\tilde{O}(k)$ work, obtaining a total of $\tilde{O}(k\cdot n^2)$ work. We get that
	\begin{equation}\label{eq:sparse}
	S(k,M)\le \tilde{O}(\min\{k\cdot n^2,M\cdot (k^{0.702}\cdot n^{1.882}+n^{2+o(1)})\})~.
	\end{equation}

	\end{comment}

	For an $n\times n$ matrix $A$, denote by $\tr_k(A)$ the matrix $A$ in which every column is truncated to contain only the smallest $k$ entries, and $\infty$ everywhere else. Clearly this operation can be executed in $\poly(\log n)$ parallel time and $\tilde{O}(n^2)$ work. For a vertex $i\in V$, let $N_k(i)$ be the set of $k$ nearest neighbors of $i$.
	\begin{claim}\label{claim:k-nn}
		Let $G$ be a weighted directed graph. For some $t \ge 1$, and $c,c' \ge 1$, let $A$ be an $n\times n$ matrix such that for every $1\le i\le n$ and every $j\in N_k(i)$, $A_{ij}$ is a $c$-approximation to $d^{(t)}_G(i,j)$, and $\infty$ for $j\notin N_k(i)$. Then, if $B$ is a $c'$-approximation to $A^T\star A$, then for each $i$ and $j\in N_k(i)$,   we have that $B_{ij}$ is a ($c\cdot c'$)-approximation to $d^{(2t)}_G(i,j)$.
	\end{claim}
	\begin{proof}
		Let $h$ be the middle vertex on the shortest path with at most $2t$ edges between $i$ and $j$ (so that there are at most $t$ edges on the sub-paths from $i$ to $h$ and from $h$ to $j$). Since $j\in N_k(i)$, the triangle inequality implies that $h\in N_k(i)$ and $j\in N_k(h)$. Thus, $A_{ih}$ (resp. $A_{hj}$) is a $c$-approximation to $d^{(t)}_G(i,h)$ (resp. $d^{(t)}_G(h,j)$).
		By definition of distance product, $(A^T\star A)_{ij}\le c\cdot d^{(t)}_G(i,h)+c\cdot d^{(t)}_G(h,j)\le c\cdot d^{(2t)}_G(i,j)$. So $B_{ij}$ is a $c\cdot c'$-approximation to $d^{(2t)}_G(i,j)$.
		(Note also that $A^T \star A)_{ij} \ge d_G^{(t)}(i,h) + d_G^{(t)}(h,j) = d_G(i,j)$.)
	\end{proof}
	
	Our algorithm to compute approximate shortest paths to $k$ nearest neighbors is done by simply computing $\log k$ times an approximate distance product, truncating each time to the smallest $k$ entries in each column. See \algref{alg:approx-k-nn}.
	(This algorithm is based on an analogous algorithm from \cite{CDKL19}, devised there in the context of the Congested Clique model.)
	
	\begin{algorithm}[ht]
		\caption{$\texttt{Approx $k$-NN}(G,\epsilon)$}\label{alg:approx-k-nn}
		\begin{algorithmic}[1]
			\STATE Let $A$ be the adjacency matrix of $G$;
			\STATE Let $R=\lceil (\log k)/\epsilon\rceil$;
			\FOR {$i$ from $1$ to $\lceil\log k\rceil$}
			\STATE Let $A$ be a ($1+1/R$)-approximation to $(\tr_k(A))^T\star\tr_k(A)$;
			\ENDFOR
			\RETURN $\tr_k(A)$;
		\end{algorithmic}
	\end{algorithm}
	
	%As the maximal weight in $G$ is $M$, the maximal entry in $A$ when executing \algref{alg:approx-k-nn} can be at most $kM$.
	Since each matrix has $m=O(nk)$ non-infinities, and there are only $O(\log k)$ iterations, the parallel time is $R\cdot\log^{O(1)}n$ and the total work, using the bound of \eqref{eq:work1} with $m=O(nk)$, is
	\[
	\tilde{O}(R\cdot  \min\{n^\omega,k^{0.702}\cdot n^{1.882}+n^{2+o(1)}\})~.
	\]

	The correctness of the algorithm follows from \claimref{claim:k-nn}, as the shortest path from a vertex $v$ to a  neighbor $u \in N_k(v)$ can have at most $k$ edges. The approximation we obtain is $(1+\frac{1}{R})^{\lceil\log k\rceil}=1+O(\epsilon)$. We remark that the truncation might actually remove the distance from $v\in V$ to some $u \in N_k(v)$, because the computed distances are only approximations, so a farther away neighbor can "replace" $u$. Denote by $N_k'(v)$ the $k$ vertices returned by \algref{alg:approx-k-nn} for $v\in V$, and note that we still obtain approximate distances every vertex in $N_k(v)$.
	
	Our algorithm can also recover the paths with approximate distances for every $i\in V$ and $j\in N_k'(i)$. This is done by applying the algorithm from \cite[Section 5]{Z02}, while executing the recursive calls in parallel.\footnote{Here is a brief sketch: Recall that we compute the witnesses for all the $O(\log k)$ distance products. Given a pair $i\in V$ and $j\in N_k'(i)$, if $W$ is the witness matrix in the last iteration of the algorithm, then there are two cases: Either $W_{ij}$ contains the middle vertex $h$ (with at most $k/2$ hops to both $i,j$) on the approximate $i-j$ path.  Then we can simply recurse in parallel on the pairs $i,h$ and $h,j$, and then concatenate the paths. Otherwise, when $W_{ij}=0$, we just return the edge $(i,j)$.}

	\begin{theorem}\label{thm:pram-approx-k-nn}
		Let $G=(V,E)$ be a weighted directed $n$-vertex  graph, and let $1 \le k \le n$ and $0<\epsilon<1$ be some parameters. Then there is a deterministic parallel algorithm that computes a ($1+\epsilon$)-approximation to all distances between any $u\in V$ and its $k$ nearest neighbors, that runs in parallel time $O((\log^{O(1)}n)/\epsilon)$, using work
		\[
		\tilde{O}( \min\{n^\omega,k^{0.702}\cdot n^{1.882}+n^{2+o(1)}\}/\epsilon)~.
		\]
		Furthermore, for each $i\in V$ and $j\in N'_k(i)$, a path achieving the approximate distance can be reported in $O(\log k)$ parallel time and work proportional to the number of edges in it.
	\end{theorem}
	Note that for $k\le n^{0.168}$ this work is $n^{2+o(1)}$, and while $k\le n^{0.698}$ the work is smaller than $n^\omega$.

	\subsection{Exact Distances}
	Here we show an efficient parallel algorithm, that given a weighted directed graph, computes the {\em exact} distances from each vertex to its $k$ nearest neighbors.
	To this end, we note that the distance product of $n\times n$ matrices with at most $k$ non-infinities in each row/column can be computed in $\log^{O(1)}n$ parallel time with $O(k n^2)$ processors. Simply assign $\tilde{O}(k)$ processors for each of the $n^2$ entries in the output matrix. They will compute the relevant inner product (over the $(\min,+)$ ring) that has at most $2k$ non-infinities, obtaining a total of $\tilde{O}(k\cdot n^2)$ work.
	The following \algref{alg:exact-k-nn} computes the exact distances.
	
	\begin{algorithm}[ht]
		\caption{$\texttt{Exact $k$-NN}(G)$}\label{alg:exact-k-nn}
		\begin{algorithmic}[1]
			\STATE Let $A$ be the adjacency matrix of $G$;
			\FOR {$i$ from $1$ to $\lceil\log k\rceil$}
			\STATE Let $A=(\tr_k(A))^T\star\tr_k(A)$;
			\ENDFOR
			\RETURN $\tr_k(A)$;
		\end{algorithmic}
	\end{algorithm}
	As there are $O(\log k)$ iterations, we get $\log^{O(1)}n$ parallel time and $\tilde{O}(k\cdot n^2)$ work. The correctness of this algorithm follows from \claimref{claim:k-nn} with $c=c'=1$, and the fact that the shortest path between $k$ nearest neighbors has at most $k$ edges. The following theorem summarizes this result.
	
	\begin{theorem}\label{thm:pram-exact-k-nn}
		Let $G=(V,E)$ be a weighted directed $n$-vertex graph, and let $1 \le k\le n$. Then there is a deterministic parallel algorithm that computes all distances between any $u\in V$ and its $k$ nearest neighbors, that runs in parallel time $\log^{O(1)}n$, using work $\tilde{O}(k\cdot n^2)$.
		
		Furthermore, for each $i\in V$ and $j\in N_k(i)$, a shortest path can be reported in $O(\log k)$ parallel time and work proportional to the number of edges in it.
	\end{theorem}
	
	We remark that for general dense graphs there is no known algorithm to compute APSP with $n^{3-\epsilon}$ work and $\poly(\log n)$ time. Hence this result is meaningful for essentially all values of $k=o(n)$.

\end{document}